%% file: 2016-ARXIV-ComputingOptimalPrefixFreeCodesWithoutSorting-Barbay.tex
\documentclass[11pt,usletter]{llncs}
\usepackage{makeidx}  

\usepackage[T1]{fontenc}
\usepackage{times} 
\usepackage{pifont}
\usepackage{epsfig}
\usepackage{algorithmic}
\usepackage{algorithm} 
\usepackage{graphicx}
\usepackage{amssymb}
\usepackage{color} 
\usepackage{fullpage}
\usepackage[bookmarks,colorlinks=false,bookmarksopen,bookmarksdepth=2, pdfauthor={J\'er\'emy Barbay}, pdfsubject={Optimal Prefix Free Codes}, pdfkeywords={Prefix Free Codes}]{hyperref}

\usepackage{version}
\excludeversion{INUTILE}
\excludeversion{TODO}
\includeversion{LONG}\excludeversion{SHORT}
\includeversion{VLONG}
\includeversion{example}
\includeversion{lproof}
\excludeversion{ALGEBRAIC}
\excludeversion{EXTERNAL}
\excludeversion{PARALLEL}
\excludeversion{BOUNDED}
\includeversion{PRACTICAL}
\excludeversion{PREVIOUSWORK}
\excludeversion{FUTUREWORK}
\excludeversion{ACKNOWLEDGEMENTS}

\begin{document}

\pagestyle{headings}  
\addtocmark{        Optimal Prefix Free Codes with Partial Sorting            }

\mainmatter              
\title{             Optimal Prefix Free Codes \\ with Partial Sorting}
\titlerunning{      Optimal Prefix Free Codes with Partial Sorting            } 
\author{J\'er\'emy Barbay}
\authorrunning{J\'er\'emy Barbay}   
%
\tocauthor{J\'er\'emy Barbay (Universdad de Chile)}
\institute{
  Departmento de Ciencias de la Computaci{\'o}n, \\
  University of Chile, \\
  \texttt{jeremy@barbay.cl}
}

\maketitle              
\begin{center} {\small \emph{(Full version)}}
\end{center}

 \begin{abstract}
We describe an algorithm computing an optimal prefix free code for $n$ unsorted positive weights in time within $O(n(1+\lg \alpha))\subseteq O(n\lg n)$, where the alternation $\alpha\in[1..n-1]$ measures the amount of sorting required by the computation.  This asymptotical complexity is within a constant factor of the optimal in the algebraic decision tree computational model, in the worst case over all instances of size $n$ and alternation $\alpha$.  Such results refine the state of the art complexity of $\Theta(n\lg n)$ in the worst case over instances of size $n$ in the same computational model, a landmark in compression and coding since 1952, by the mere combination of van Leeuwen's algorithm to compute optimal prefix free codes from sorted weights (known since 1976), with Deferred Data Structures to partially sort a multiset depending on the queries on it (known since 1988).
 \end{abstract}

\begin{center}
  \begin{minipage}{.9\textwidth}
    \noindent{\bf Keywords:} 
Deferred Data Structure,
Huffman,
Median,
Optimal Prefix Free Codes,
van Leeuwen.
  \end{minipage}
\end{center}

\input{notations.tex}
\input{adaptivePFC.tex}

\bibliographystyle{abbrv}
\bibliography{/home/jbarbay/EverGoing/WebSite/Studies/ScientificArticles/biblio-Barbay,/home/jbarbay/EverGoing/WebSite/Publications/publications-Barbay}

\end{document}

%% file: notations.tex
\providecommand{\id}[1]{\ensuremath{\mathit{#1}}}
\let\idit=\id
\providecommand{\idbf}[1]{\ensuremath{\mathbf{#1}}}
\providecommand{\idrm}[1]{\ensuremath{\mathrm{#1}}}
\providecommand{\idtt}[1]{\ensuremath{\mathtt{#1}}}
\providecommand{\idsf}[1]{\ensuremath{\mathsf{#1}}}
\providecommand{\idcal}[1]{\ensuremath{\mathcal{#1}}}  
\providecommand{\keypoint}[1]{{#1}} 
\providecommand{\llg}{\ensuremath{\log_2\log_2}}
\providecommand{\etal}{~\emph{et al.}}
\providecommand{\Oh}{\ensuremath{O}}
\providecommand{\entropy}{{\ensuremath\cal H}}
\providecommand{\redundancy}{{\ensuremath \cal R}}
\providecommand{\alternation}{{\ensuremath \alpha}}
\providecommand{\inputSize}{{\ensuremath n}}
\providecommand{\alphabetSize}{{\ensuremath k}}
\providecommand{\nbMessages}{{\ensuremath |T|}}
\providecommand{\nbWeights}{{\ensuremath n}}
\providecommand{\nbClusters}{\delta}
\providecommand{\nbDistinctWeights}{r}
\providecommand{\distribution}{{\ensuremath \Delta}}
\providecommand{\weightVector}{{\ensuremath W}}
\providecommand{\weight}[1]{{\ensuremath W[#1]}}
\providecommand{\weightVar}{{\ensuremath x}}
\providecommand{\sumWeights}{{\ensuremath U}}
\providecommand{\nbSymbols}{{\ensuremath D}}
\providecommand{\codeLength}[1]{{\ensuremath L[#1]}}
\providecommand{\nbPages}{{\ensuremath n}}
\providecommand{\nbCodeLengths}{{\ensuremath k}}
\providecommand{\minimalNbCodeLengths}{{\ensuremath \kappa}}
\providecommand{\partialSum}{{\ensuremath S}}
\providecommand{\rankVar}{{\ensuremath\rho}}
\providecommand{\median}{{\ensuremath\idtt{median}}}
\providecommand{\select}{{\ensuremath\idtt{select}}}
\providecommand{\partition}{{\ensuremath\idtt{Partition}}}
\providecommand{\rank}{{\ensuremath\idtt{rank}}}
\providecommand{\offset}{{\ensuremath\idtt{offset}}}
\providecommand{\size}{{\ensuremath\idtt{size}}}
\providecommand{\MinWeight}{{\ensuremath\idtt{min}}}
\providecommand{\MaxWeight}{{\ensuremath\idtt{max}}}
\providecommand{\LogSort}{{\ensuremath\idtt{LogSort}}}
\providecommand{\LogArray}{{\ensuremath\idtt{LIP}}}
\providecommand{\TopDown}{{\ensuremath\idtt{TopDown}}}
\providecommand{\DAryTopDown}{{\ensuremath\idtt{DAryTopDown}}}
\providecommand{\CodeLengthVariation}{{\ensuremath\idtt{Variation}}}
\providecommand{\NbWeights}{{\ensuremath\idtt{NbWeights}}}
\providecommand{\NbClusters}{{\ensuremath\idtt{NbClusters}}} 
\providecommand{\NbRoots}{{\ensuremath\idtt{NbRoots}}}
\providecommand{\NbTuples}{{\ensuremath\idtt{NbTuples}}}
\providecommand{\NbSymbols}{{\ensuremath\idtt{NbSymb}}}
\providecommand{\MinCodeLength}{{\ensuremath\idtt{MinLen}}}
\providecommand{\SplittingPoint}{{\ensuremath\idtt{SplitPoint}}}
\providecommand{\ClusterSize}{{\ensuremath\idtt{ClSize}}}
\providecommand{\MaxNbNodes}{{\ensuremath\idtt{MaxNbNodes}}}
\providecommand{\MinLeafWeight}{{\ensuremath\idtt{MinLeaf}}} 
\providecommand{\SumWeightsNodes}{{\ensuremath\idtt{SumWeights}}} 
\providecommand{\ClusterStart}{{\ensuremath\idtt{ClStart}}}
\providecommand{\ClusterThreshold}{{\ensuremath\idtt{ClThreshold}}}
\providecommand{\NbPairs}{{\ensuremath\idtt{NbPairs}}}
\providecommand{\PartialSum}{{\ensuremath\idtt{PartialSum}}}
\providecommand{\toto}{{\ensuremath\idtt{}}}
\providecommand{\lengthBound}{{\ensuremath B}}


%% file: adaptivePFC.tex
\section{Introduction}
%
Given $\nbWeights$ \begin{LONG}positive\end{LONG} weights $\weight{1..\nbWeights}$ coding\begin{LONG}\footnote{We note $[i..j]=\{i,i+1,\ldots,j\}$ the integer range from $i$ to $j$, and $A[i..j]=\{A[i],A[i+1],\ldots,A[j]\}$ the set of values of an array $A$ indexed within this range.}\end{LONG} for the frequencies
\begin{LONG}
$\left\{{\weight{i}}/{\sum_{j=1}^\nbWeights\weight{j}}\right\}_{i\in[1..\nbWeights]}$
\end{LONG}
\begin{SHORT}
$\{{\weight{i}}/{\sum_{j=1}^\nbWeights\weight{j}}\}_{i\in[1..\nbWeights]}$
\end{SHORT}
of $\nbWeights$ messages\begin{LONG}\footnote{We use the terminology of \emph{messages} for the input and \emph{symbols} for the output, as introduced by Huffman~\cite{1952-IRE-AMethodForTheInstructionOfMinimumRedundancyCodes-Huffman}, which should not be confused with other terminologies found in the literature, of \emph{input symbols}, \emph{letters} or \emph{words} for the input and \emph{output symbols} or \emph{bits} in the binary case.}\end{LONG}, and a number $\nbSymbols$ of output symbols,
an \textsc{Optimal Prefix Free Code}~\cite{1952-IRE-AMethodForTheInstructionOfMinimumRedundancyCodes-Huffman} is a set of $\nbWeights$ code strings on alphabet $[1..\nbSymbols]$, of variable lengths $\codeLength{1..\nbWeights}$ and such that no string is prefix of another, and the average length of a code is minimized (i.e. $\sum_{i=1}^\nbWeights\codeLength{i}\weight{i}$ is minimal).
\begin{LONG}
The particularity of such codes is that even though the code strings assigned to the messages can differ in lengths (assigning shorter ones to more frequent messages yields compression to $\sum_{i=1}^\nbWeights\codeLength{i}\weight{i}$ symbols), the prefix free property insures a non-ambiguous decoding.
 
\end{LONG}
\begin{LONG}
Such optimal codes, known since
1952~\cite{1952-IRE-AMethodForTheInstructionOfMinimumRedundancyCodes-Huffman},
are used in ``\emph{all the mainstream compression  formats}''~\cite{2006-IEEE-LowPowerHuffmanCodingForHighPerformanceDataTransmission-Chen}
(e.g.  \texttt{PNG}, \texttt{JPEG}, \texttt{MP3}, \texttt{MPEG}, \texttt{GZIP} and \texttt{PKZIP}).
The concept is ``\emph{one of the {fundamental ideas that people} in computer science and data communications {are using all the time}}'' (Knuth~\cite{2010-BOOK-DiscreteMathematics-Chandrasekaran}), and the code itself is ``\emph{one of the {enduring techniques of data compression}. It was used in the venerable PACK compression program, authored by Szymanski in 1978, and {remains no less popular today}}'' (Moffat\etal~\cite{1997-IEEE-OnTheImplementstionOfMinimumRedundsncyPrefixCodes-MoffatTurpin} in 1997).
\end{LONG}

\begin{LONG}
\subsection{Previous works}\label{sec:previous-works}
\end{LONG}

%
Any prefix free code can be computed in linear time from a set of code lengths satisfying the Kraft inequality $\sum_{i=1}^\nbWeights\nbSymbols^{-\codeLength{i}}\leq1$.
The original description of the code by Huffman~\cite{1952-IRE-AMethodForTheInstructionOfMinimumRedundancyCodes-Huffman} yields a heap-based algorithm performing $O(\nbWeights\log\nbWeights)$ algebraic operations, using the bijection between $\nbSymbols$-ary prefix free codes and $\nbSymbols$-ary cardinal trees~\cite{2012-Book-GraphAlgorithms-EvenEven}.
This complexity is asymptotically optimal for any constant value of $\nbSymbols$ in the algebraic decision tree model, in the worst case over instances composed of $\nbWeights$ positive weights\begin{LONG}, as computing the optimal binary prefix free code for the weights $\weight{0,\ldots,\nbSymbols\nbWeights}=\{\nbSymbols^{x_1},\ldots,\nbSymbols^{x_1},\nbSymbols^{x_2},\ldots,\nbSymbols^{x_2},\ldots,\nbSymbols^{x_\nbWeights},\ldots,\nbSymbols^{x_\nbWeights}\}$ is equivalent to sorting the positive integers $\{x_1,\ldots,x_\nbWeights\}$
\end{LONG}. We consider here only the binary case, where $D=2$. 
Not all instances require the same amount of work to compute an optimal code \begin{LONG} (see Table~\ref{tab:previousResults} for a partial list of relevant results)\end{LONG}:
\begin{itemize}

\item When the weights are given in sorted order, van Leeuwen~\cite{1976-ICALP-OnTheConstructionOfHuffmanTrees-Leeuwen} showed that an optimal code can be computed using within $O(\nbWeights)$ algebraic operations.

\item When the weights consist of $r\in[1..n]$ distinct values and are given in a sorted, compressed form, Moffat and Turpin~\cite{1998-TIT-EfficientConstructionOfMinimumRedundancyCodesForLargeAlphabets-MoffatTurpin} showed how to compute an optimal code using within $O(r(1+\log(\nbWeights/r)))$ algebraic operations, which is often sublinear in $n$.

\item 
  In the case where the weights are given unsorted, Belal\etal~\cite{2006-STACS-DistributionSensitiveConstructionOfMinimumRedundancyPrefixCodes-BelalElmasry,2006-IEEE-VerificationOfMinimumRedundancyPrefixCodes-BelalElmasry}  described
  \begin{LONG}
  several families of instances for which an optimal prefix free code can be computed in linear time, along with 
  \end{LONG}
  an algorithm claimed to perform $O(\nbCodeLengths\nbWeights)$ algebraic operations, in the worst case over instances formed by $\nbWeights$ weights such that there is an optimal binary prefix free code with $\nbCodeLengths$ distinct code lengths\begin{LONG}\footnote{Note that $k$ is not uniquely defined, as for a given set of weights there can exist several optimal prefix free codes varying in the number of distinct code lengths used.}\end{LONG}.
  \begin{LONG}
  This complexity was later downgraded to $O(16^k n)$ in an extended version\cite{2005-ARXIV-DistributionSensitiveConstructionOfMinimumRedundancyPrefixCodes-BelalElmasry} of their article. Both results are better than the state of the art when $k$ is finite, but worse when $k$ is larger than $\log n$.
  \end{LONG}

\end{itemize}
\begin{LONG}
\begin{table}
\centering
\begin{tabular}{cp{4cm}llcll}
Year & Name                          & Time                           & Space           & Ref.                                                                                           & Note                         \\ \hline
1952 & Huffman                       & $O(\nbWeights\log \nbWeights)$ & $O(\nbWeights)$ & \cite{1952-IRE-AMethodForTheInstructionOfMinimumRedundancyCodes-Huffman}                       & original                     \\ \hline
1976 & van Leeuwen                   & $O(\nbWeights)$                & $O(\nbWeights)$ & \cite{1976-ICALP-OnTheConstructionOfHuffmanTrees-Leeuwen}                                      & Sorted Input                 \\ \hline
1995 & Moffat and Katajainen         & $O(\nbWeights)$                & $O(1)$          & \cite{1995-WADAS-InPlaceCalculationOfMinimumRedundancyCodes-MoffatKatajainen}                  & Sorted Input                 \\ \hline
1998 & Moffat and Turpin             & $O(r(1+\log(\nbWeights/r)))$   & "efficient"     & \cite{1998-TIT-EfficientConstructionOfMinimumRedundancyCodesForLargeAlphabets-MoffatTurpin}    & Compressed Input/Output      \\ \hline
2006 & Belal and Elmasry             & $O(\log^{2k-1} \nbWeights)$    &                 & \cite{2006-STACS-DistributionSensitiveConstructionOfMinimumRedundancyPrefixCodes-BelalElmasry} & Sorted Input                 \\ 
2006 & Belal and Elmasry             & $O(k\nbWeights)$ claimed       & $O(\nbWeights)$ & \cite{2006-STACS-DistributionSensitiveConstructionOfMinimumRedundancyPrefixCodes-BelalElmasry} & $k$ distinct code lengths    \\
2006 & Belal and Elmasry             & $O(16^k\nbWeights)$ proved     & $O(\nbWeights)$ & \cite{2005-ARXIV-DistributionSensitiveConstructionOfMinimumRedundancyPrefixCodes-BelalElmasry} & $k$ distinct code lengths    \\ \hline
2016 & {{Grouping-Docking-Mixing}} & $O(n(1+\log\alpha))$           & $O(\nbWeights)$ & [here]                                                                                     & $\alpha=|S|_{EI}\in[1..n-1]$ \\ \hline
\end{tabular}
\caption{A selection of results on the computational complexity of optimal prefix free codes.  %
$k$ is the number of distinct codelengths produced. %
$\alpha=|S|_{EI}\in[1..n-1]$ is a difficulty measure, the number of alternation between External nodes and Internal nodes in an execution of van Leeuwen~\cite{1976-ICALP-OnTheConstructionOfHuffmanTrees-Leeuwen}'s algorithm.
Note that there can be various optimal codes for any given set of weights, each with a distinct number of distinct code lengths $k$.
}
\label{tab:previousResults}
\end{table}
\end{LONG}

\begin{LONG}
\subsection{Contributions}\label{sec:contributions}
\end{LONG}

\begin{LONG}
In the context described above, various questions are left unanswered, from the confirmation of the existence of an algorithm running in time $O(16^k\nbWeights)$ or $O(k\nbWeights)$, to the existence of an algorithm taking advantage of small values of both $n$ and $k$, less trivial than running two algorithms in parallel and stopping both whenever one computes the answer.
\end{LONG}
%
Given $\nbWeights$ positive integer weights, \emph{can we compute an optimal binary prefix free code in time better than $O(\min\{\nbCodeLengths\nbWeights,\nbWeights\log\nbWeights\})$ in the algebraic model?  }
We answer in the affirmative for many classes of instances, identified by the alternation measure $\alpha$ defined in Section \ref{sec:alternation}:
\begin{theorem}
Given $\nbWeights$ positive weights of alternation $\alternation\in[1..n-1]$, there is an algorithm which computes an optimal binary prefix free code using within $O(n(1{+}\log \alpha))$ $\subseteq O(n\lg n)$ algebraic instructions, and this complexity is asymptotically optimal among all algorithms in the algebraic decision tree  computational model in the worst case over instances of size $n$ and alternation $\alpha$.
\end{theorem}
\begin{proof}
\begin{LONG}
We describe in Lemma~\ref{result:dds} a deferred data structure which supports $q$ queries of type \texttt{rank}, \texttt{select} and \texttt{partialSum} in time within $O(n(1+\lg q))$, all within the algebraic computational model, and describe in Section~\ref{sec:algorithm} an algorithm using such a data structure to compute optimal prefix free codes given an unsorted input.
\end{LONG}
We show in Lemma~\ref{result:lowerBoundWorstCase} that any algorithm $A$ in the algebraic computational model performs within $\Omega(n\lg \alpha$) algebraic operations in the worst case over instances of size $\nbWeights$ and alternation~$\alternation$.
We show in Lemma~\ref{result:upper-boundQueries} that the \texttt{GDM} algorithm, a variant of the van Leeuwen's algorithm~\cite{1976-ICALP-OnTheConstructionOfHuffmanTrees-Leeuwen}, modified to use the deferred data structure from Lemma~\ref{result:dds}, performs $q\in O(\alpha(1+\lg\frac{n-1}{\alpha}))$ such queries, which yields in Corollary~\ref{result:upper-boundOperations} a complexity within $O(n(1{+}\log \alpha) + \alpha(\lg n)(\lg\frac{n}{\alpha}))$, all within the algebraic computational model. 
As $\alpha\in[1..n{-}1]$ and $O(\alpha(\lg n)(\lg\frac{n}{\alpha}))\subseteq O(n(1{+}\log \alpha))$ for this range (Lemma~\ref{result:asymptotics}), the optimality ensues.  \qed
\end{proof}

\begin{LONG}
When $\alpha$ is at its maximal (i.e. $\alpha = n{-}1$), this complexity matches the tight computational complexity bound of $\Theta(n\lg n)$ for algebraic algorithms in the worst case over all instances of size $n$.
When $\alpha$ is substantially smaller than $n$ (e.g. $\alpha\in O(\lg n)$), the \texttt{GDM} algorithm performs within $o(n\lg n)$ operations, down to linear in $n$ for finite values of~$\alpha$.
\end{LONG}

We discuss our solution in Section~\ref{sec:solution} in three parts: the intuition behind the general strategy in Section \ref{sec:general-intuition}, the deferred data structure which maintains a partially sorted list of weights while supporting \texttt{rank}, \texttt{select} and \texttt{partialSum} queries in Section~\ref{sec:dds}, and the algorithm which uses those operators to compute an optimal prefix free code in Section~\ref{sec:algorithm}.  Our main contribution consists in the analysis of the running time of this solution, described in Section~\ref{sec:analysis}: the formal definition of the parameter of the analysis in Section~\ref{sec:alternation}, the upper bound in Section \ref{sec:upper-bound} and the matching lower bound in Section \ref{sec:lower-bound}.  We conclude with a comparison of our results with those from Belal\etal~\cite{2006-STACS-DistributionSensitiveConstructionOfMinimumRedundancyPrefixCodes-BelalElmasry} in Section~\ref{sec:discussion}.

\section{Solution} \label{sec:solution}

\begin{LONG}
The solution that we describe is a combination of two results: some results about deferred data structures for multisets, which support queries in a ``lazy'' way; and some results about optimal prefix free codes themselves, about the relation between the computational cost of sorting a set of positive integers and the computational cost of computing an optimal prefix free code for the corresponding frequency distribution.  We describe the general intuition of our solution in Section \ref{sec:general-intuition}, the deferred data structure in Section~\ref{sec:dds}, and the algorithm in Section~\ref{sec:algorithm}.
\end{LONG}

\subsection{General Intuition}\label{sec:general-intuition}

Observing that the algorithm suggested by Huffman~\cite{1952-IRE-AMethodForTheInstructionOfMinimumRedundancyCodes-Huffman} always creates the internal nodes in increasing order of weight, van Leeuwen~\cite{1976-ICALP-OnTheConstructionOfHuffmanTrees-Leeuwen} described an algorithm to compute optimal prefix free codes in linear time when the input (i.e. the weights of the external nodes) is given in sorted order.

A close look at the execution of van Leeuwen's algorithm~\cite{1976-ICALP-OnTheConstructionOfHuffmanTrees-Leeuwen} reveals a sequence of \texttt{sequential searches} for the insertion rank $r$ of the weight of an internal node in the list of weights of external nodes.  Such sequential search could be replaced by a more efficient search algorithm\begin{LONG} in order to reduce the number of comparisons performed (e.g. a \texttt{doubling search}~\cite{1976-IPL-AnAlmostOptimalAlgorithmForUnboundedSearching-BentleyYao} would find such a rank $r$ in $2\lceil\log_2 r\rceil$ comparisons)\end{LONG}.

\begin{example}
Consider an instance of the optimal prefix free code problem formed by $\nbWeights$ \emph{sorted} positive weights $\weight{1..\nbWeights}$ such that the first internal node created is bigger than the largest weight (i.e. $W[1]+W[2] > W[n]$).  On such an instance, van Leeuwen's algorithm~\cite{1976-ICALP-OnTheConstructionOfHuffmanTrees-Leeuwen} starts by performing $n-2$ comparisons in the equivalent of a \texttt{sequential search} in $W$ for $W[1]{+}W[2]$: a \texttt{binary search} would perform $\lceil\log_2 n\rceil$ comparisons instead.
\end{example}

Of course, any algorithm must access (and sum) each weight at least once in order to compute an optimal prefix free code for the input, so that reducing the number of comparisons does not reduce the running time of van Leeuwen's algorithm on a sorted input. Our claim is that in the case where the input is not sorted, the computational cost of optimal prefix free codes on instances where van Leeuwen performs long sequential searches can be greatly reduced. We define the ``van Leeuwen signature'' of an instance as a first step to characterize  such instances:

\begin{definition}
Given an instance of the optimal prefix free code problem formed by $\nbWeights$ positive weights $\weight{1..\nbWeights}$, its \emph{van Leeuwen signature} ${\cal S}(W)\in \{E,I\}^{2n-1}$ is a string of length $2n-1$ over the alphabet $\{E,I\}$ (where $E$ stands for ``External'' and $I$ for ``Internal'') marking, at each step of the algorithm described by van Leeuwen~\cite{1976-ICALP-OnTheConstructionOfHuffmanTrees-Leeuwen}, whether an external or internal node is chosen as the minimum (including the last node returned by the algorithm, for simplicity).
\end{definition}

\begin{example}
Given the sorted array $W=\begin{array}{|*{8}{c|}}\hline 1&2&3&4&5&5&6&7 \\ \hline\end{array}$ of length $8$, 
its \emph{van Leeuwen signature} is of length $15$, starts with $EE$ and finishes with $I$: ${\cal S}(W) = \mathtt{EEEIEEEEIEIIIII}$.
\end{example}

The analysis described in Section~\ref{sec:analysis} is based on the number of blocks formed only of $E$ in the van Leeuwen signature of the instance $S$. We can already show some basic properties of this measure:

\begin{lemma}
Given the van Leeuwen signature $S$ of $n$ unsorted positive weights $W[1..n]$,
 $|S|_E = n$;
 $|S|_I = n-1$;
 $|S| = 2n-1$;
$S$ starts with two $E$;
$S$ finishes with one $I$;
 $|S|_{EI}=|S|_{IE}+1$; 
$|S|_{EI}\in[1..n-1]$.
\end{lemma}
\begin{lproof}
The three first properties are simple consequences of basic properties on binary trees.
$S$ starts with two $E$ as the first two nodes paired are always external.
$S$ finishes with one $I$ as the last node returned is always (for $n>1$) an internal node.
The two last properties are simple consequences of the fact that $S$ is a binary string starting with an $E$ and finishing with an $I$. 
\end{lproof}

Instances with very few blocks of $E$ are easier to solve than instances with many such blocks. For instance, an instance $W$ of length $n$ such that its signature ${\cal S}(W)$ is composed of a single run of $n$ $E$s followed by a single run of $n-1$ $I$s can be solved in linear time, and in particular without sorting the weights: it is enough to assign the codelength $l=\lfloor \log_2 n\rfloor$ to the $n-2^l$ largest weights and the codelength $l+1$ to the $2^l$ smallest weights. Separating those weights is a simple \texttt{select} operation, supported by the data structures described in the following section.

\subsection{Partial Sum Deferred Data Structure}
\label{sec:dds}

Given a \textsc{Multiset} $W[1..n]$ on alphabet $[1..\sigma]$ of size $n$, Karp\etal~\cite{1988-JC-DeferredDataStructuring-KarpMotwaniRaghavan} defined the first deferred data structure supporting for all $x\in[1..\sigma]$ and $r\in[1..n]$ queries such as \texttt{rank}$(x)$, the number of elements which are strictly smaller than $x$ in $W$; and \texttt{select}$(r)$, the value of the $r$-th smallest value (counted with multiplicity) in $W$.  Their data structure supports $q$ queries in time within $O(n(1+\lg q))$, all in the comparison model.
\begin{LONG}
To achieve this results, it partially sorts its data in order to minimize the computational cost of future queries, but avoids sorting all of the data if the queries don't require it: the queries have then become operators (they modify the data). Note that whereas the running time of each individual query depends on the state of the data, the answer to each query is independent of the state of the data.
\end{LONG}

Karp\etal's data structure~\cite{1988-JC-DeferredDataStructuring-KarpMotwaniRaghavan} supports only \texttt{rank} and \texttt{select} queries in the comparison model, whereas the computation of optimal prefix free codes requires to sum pairs of weights from the input, and the algorithm that we propose in Section~\ref{sec:algorithm} requires to sum weights from a range in the input. Such requirement can be reduced to \texttt{partialSum} queries. Whereas such queries have been defined in the literature, we define them here in a way that depends only on the content of the \textsc{Multiset} (as opposed to a definition dpending on the order in which it is given), so that it can be generalized to deferred data structures.

\begin{definition}
Given $n$ unsorted positive weights $W[1..n]$, a \texttt{Partial Sum} data structure supports the following queries:
\begin{LONG}
\begin{itemize}
\item \texttt{rank}$(x)$, the number of elements which are strictly smaller than $x$ in $W$;
\item \texttt{select}$(r)$, the value of the $r$-th smallest value (counted with multiplicity) in $W$;
\item \texttt{partialSum}$(r)$, the sum of the $r$ smallest elements (counted with multiplicity) in $W$.
  \end{itemize}
  \end{LONG}
  \begin{SHORT}
  \texttt{rank}$(x)$, the number of elements which are strictly smaller than $x$ in $W$; \texttt{select}$(r)$, the value of the $r$-th smallest value (counted with multiplicity) in $W$; \texttt{partialSum}$(r)$, the sum of the $r$ smallest elements (counted with multiplicity) in $W$.
  \end{SHORT}
\end{definition}

\begin{example}
Given the array $A=\begin{array}{|*{8}{c|}}\hline5&3&1&5&2&4&6&7 \\ \hline\end{array}$, \texttt{rank}$(5)=4$, \texttt{select}$(6)=5$, and \texttt{partialSum}$(2)=3$.
\end{example}

We describe below how to extend Karp\etal's deferred data structure~\cite{1988-JC-DeferredDataStructuring-KarpMotwaniRaghavan}, which supports \texttt{rank} and \texttt{select} queries on \textsc{Multisets}, in order to add  the support for \texttt{partialSum} queries, with an amortized running time within a constant factor of the original asymptotic time. Note that the data structure is not performing any more in the comparison model, but rather in the algebraic decision tree model, since it performs algebraic operations (additions) on the elements of the \textsc{Multiset}:

\begin{lemma} \label{result:dds}
Given $n$ unsorted positive weights $W[1..n]$, there is a \texttt{PartialSum} Deferred Data Structure which supports $q$ operations of type \texttt{rank}, \texttt{select} and \texttt{partialSum} in time within $O(n(1+\lg q)+q(1+\log n))$, all within the algebraic decision tree computational model.
\end{lemma}

\begin{lproof}
Karp\etal~\cite{1988-JC-DeferredDataStructuring-KarpMotwaniRaghavan} described a deferred data structure which supports the \texttt{rank} and \texttt{select} queries (but not \texttt{partialSum} queries).  It is based on median computations and $(2,3)$-trees, and performs $q$ queries on $n$ values in time within $O(n(1+\lg q)+q(1+\log n))$, all within the algebraic computational model. We describe below how to modify in a simple way their data structure so that to support \texttt{partialSum} queries with asymptotically negligible additional cost.
At the initialization of the data structure, compute the $n$ partial sums corresponding to the $n$ positions of the unsorted array. After each median computation and partitioning in a \texttt{rank} or \texttt{select} query, recompute the partial sums on the range of values newly partitioned, adding only a constant factor to the cost of the query. When answering a \texttt{partialSum} query, perform a \texttt{select} query and then return the value of the partial sum corresponding to the value by the \texttt{select} query: the asymptotic complexity is within a constant factor of the one described by Karp\etal~\cite{1988-JC-DeferredDataStructuring-KarpMotwaniRaghavan}.\qed
\end{lproof}

\begin{LONG}
Barbay\etal~\cite{2013-ESA-OnlineRankSelect-BarbayGuptaJoRaoSorenson} further improved Karp\etal's result~\cite{1988-JC-DeferredDataStructuring-KarpMotwaniRaghavan} with a simpler data structure (a single binary array) and a finer analysis taking into account the gaps between the position hit by the queries. Barbay\etal's results~\cite{2013-ESA-OnlineRankSelect-BarbayGuptaJoRaoSorenson} can similarly be augmented in order to support \texttt{partialSum} queries while increasing the computational complexity by only a constant factor. This result is not relevant to the analysis described in Section~\ref{sec:analysis}.
\end{LONG}

\begin{LONG}
Such a deferred data structure is sufficient to simply execute van Leeuwen's algorithm~\cite{1976-ICALP-OnTheConstructionOfHuffmanTrees-Leeuwen} on an unsorted array of positive integers, but would not result in an improvement in the computational complexity: van Leeuwen's algorithm~\cite{1976-ICALP-OnTheConstructionOfHuffmanTrees-Leeuwen} is simply performing $n$ \texttt{select} operations on the input, effectively sorting the unsorted array.
\end{LONG}

We describe in the next section an algorithm which uses the deferred data structure described above to batch the operations on the external nodes, and to defer the computation of the weights of some internal nodes to later, so that for many instances the input is not completely sorted at the end of the execution, which reduces the execution cost.


\subsection{Algorithm ``\texttt{Grouping-Docking-Mixing}'' (\texttt{GDM})}\label{sec:algorithm}

\begin{TODO}
We describe the pseudo-code for the \texttt{Grouping-Docking-Mixing} algorithm (\texttt{\texttt{GDM}} for short) in Algorithm \ref{alg:aaa}.
\begin{algorithm}
\caption{\texttt{\texttt{GDM} Algorithm}}
\label{alg:aaa}
\textbf{Input}: a \texttt{PartialSum} deferred data structure $W[1..n]$, initialized with $n$ unsorted positive weights.
\\
\textbf{Output}: a tree representing an optimal prefix free code for the frequencies in $W[1..n]$.
\hrule
\begin{minipage}{.49\textwidth}
\begin{algorithmic}
\IF{ n==1 }
\STATE return $W$
\ENDIF
\STATE Initialize the \texttt{Partial Sum} deferred data structure with $W$; \\
\STATE \texttt{nbExternalProcessed = 2};
\STATE \texttt{currentMinExternal = select(3)};
\STATE \texttt{currentMinInternal = partialSum(2)};
\STATE \texttt{nbInternals = 1};
\STATE \texttt{Internals} = [(partialSum(2),1,2)];
\WHILE{ \texttt{nbExternalProcessed < n}}
  \STATE \texttt{r = rank(currentMinInternal)}; 
  \STATE (...)
\ENDWHILE  
\STATE {\bf return};
\end{algorithmic}
\end{minipage}
\end{algorithm}

\end{TODO}
There are five main phases in the \texttt{\texttt{GDM}} algorithm: the \emph{Initialization}, three phases (\emph{Grouping}, \emph{Docking} and \emph{Mixing}, hence the name ``\texttt{GDM}'' of the algorithm) inside a loop running until only internal nodes are left to process, and the \emph{Conclusion}:
\begin{itemize}
\item In the \emph{Initialization} phase, initialize the \texttt{Partial Sum} deferred data structure with the input, and initialize the first internal node by pairing the two smallest weights of the input.
\item In the \emph{Grouping} phase, detect and group the weights smaller than the smallest internal node: this corresponds to a run of consecutive $E$ in the van Leeuwen signature of the instance.
\item In the \emph{Docking} phase, pair the consecutive \emph{positions} of those weights (as opposed to the weights themselves, which can be reordered by future operations) into internal nodes, and pair  those internal nodes until the weight of at least one such internal node becomes equal or larger than the smallest remaining weight: this corresponds to a run of consecutive $I$ in the van Leeuwen signature of the instance.
\item In the \emph{Mixing} phase, rank the smallest unpaired weight among the weights of the available internal nodes: this corresponds to an occurrence of $IE$ in the van Leeuwen signature of the instance.
\item In the \emph{Conclusion} phase, with $i$ internal nodes left to process,  assign codelength $l=\lfloor \log_2 i\rfloor$ to the $i-2^l$ largest ones and  codelength $l{+}1$ to the 
$2^l$ smallest ones: this corresponds to the last run of consecutive $I$ in the van Leeuwen signature of the instance.
\end{itemize}

The algorithm and its complexity analysis distinguish two types of internal nodes: \emph{pure} nodes, which descendants were all paired during the same \emph{Grouping} phase; and \emph{mixed} nodes, which either is the ancestor of such a \emph{mixed} node, or pairs a \emph{pure} internal node with an external node, or pairs two \emph{pure} internal nodes produced at distinct phases of the algorithm.  The distinction is important as the algorithm computes the weight of any \emph{mixed} node at its creation (potentially generating several data structure operations), whereas it defers the computation of the weight of some \emph{pure} nodes to later.

Before describing each phase more in detail, it is important to observe the following invariant of the algorithm:
\begin{lemma} \label{result:allInternalsWithinAFactorOfTwo}
Given an instance of the optimal prefix free code problem formed by $\nbWeights>1$ positive weights $\weight{1..\nbWeights}$,
between each phase of the algorithm, 
all unpaired internal nodes have weight within a constant factor of two (i.e. the maximal weight of an unpaired internal node is strictly smaller than the minimal weight of an unpaired internal node).
\end{lemma}

We now proceed to describe each phase in more details:

\paragraph{Initialization:} 
Initialize the \texttt{Partial Sum} deferred data structure;
compute the weight $\idtt{currentMinInternal}$ of the first internal node through the operation $\idtt{partialSum}(2)$ (the sum of the two smallest weights); 
create this first internal node as a node of weight $\idtt{currentMinInternal}$ and children $1$ and $2$ (the positions of the first and second weights, in any order);
compute the weight \idtt{currentMinExternal} of the first unpaired weight (i.e. the first available external node) by the operation $\idtt{select}(3)$;
setup the variables $\idtt{nbInternals}=1$ and $\idtt{nbExternalProcessed}=2$.

\paragraph{Grouping:}
Compute the position $r$ of the first unpaired weight which is larger than the smallest unpaired internal node, through the operation \texttt{rank} with parameter $\idtt{currentMinInternal}$;
pair the $((r-\idtt{nbExternalProcessed})$ modulo $2)$ indices to form $\lfloor\frac{r-\idtt{nbExternalProcessed}}{2}\rfloor$ \emph{pure} internal nodes;
if the number $r-\idtt{nbExternalProcessed}$ of unpaired weights smaller than the first unpaired internal node is odd, select the $r$-th weight through the operation $\idtt{select}(r)$, compute the weight of the first unpaired internal node, compare it with the next unpaired weight, to form one \emph{mixed} node by combining the minimal of the two with the extraneous weight.

\paragraph{Docking:}
Pair all internal nodes by batches (by Lemma~\ref{result:allInternalsWithinAFactorOfTwo}, their weights are all within a factor of two, so all internal nodes of a generation are processed before any internal node of the next generation);
after each batch, compare the weight of the largest such internal node (compute it through $\idtt{partialSum}$ on its range if it is a \emph{pure} node, otherwise it is already computed) with the first unpaired weight: if smaller, pair another batch, and if larger, the phase is finished.

\paragraph{Mixing:}
Rank the smallest unpaired weight among the weights of the available internal nodes, by a doubling search starting from the begining of the list of internal nodes. For each comparison, if the internal node's weight is not already known, compute it through a \texttt{partialSum} operation on the corresponding range (if it is a \emph{mixed} node, it is already known). If the number $r$ of internal nodes of weight smaller than the unpaired weight is odd, pair all but one, compute the weight of the last one and pair it with the unpaired weight. If $r$ is even, pair all of the $r$ internal nodes of weight smaller than the unpaired weight, compare the weight of the next unpaired internal node with the weight of the next unpaired external node, and pair the minimum of the two with the first unpaired weight.
If there are some unpaired weights left, go back to the \emph{Grouping} phase, otherwise continue to the \emph{Conclusion} phase.

\paragraph{Conclusion:}
There are only internal nodes left, and their weights are all within a factor of two from each other. 
Pair the nodes two by two in batch as in the \emph{Docking} phase, computing the weight of an internal node only when the number of internal nodes of a batch is odd.

The combination of those phases forms the \texttt{GDM} algorithm, which computes an optimal prefix free code given an unsorted sets of positive integers.
\begin{LONG}
\begin{lemma}
The tree returned by the \texttt{GDM} algorithm describes an optimal prefix free code for its input.
\end{lemma}
\end{LONG}
In the next section, we analyze the number $q$ of \texttt{rank}, \texttt{select} and \texttt{partialSum} queries performed by the \texttt{GDM} algorithm, and deduce from it the complexity of the algorithm in term of algebraic operations.

\section{Analysis}\label{sec:analysis}

The \texttt{GDM} algorithm runs in time within $O(n\lg n)$ in the worst case over instances of size $n$ (which is optimal (if not a new result) in the algebraic decision tree model), but much faster on instances with few blocks of consecutive $E$s in the van Leeuwen signature of the instance. We formalize this concept by defining the \emph{alternation} $\alpha$ of the instance in Section~\ref{sec:alternation}. We then proceed in Section~\ref{sec:upper-bound} to show upper bounds on the number of queries and operations performed by the \texttt{GDM} algorithm in the worst case over instances of fixed size $n$ and alternation $\alpha$. We finish in Section~\ref{sec:lower-bound} with a matching lower bound for the number of operations performed.

\subsection{Alternation $\alpha(W)$} \label{sec:alternation}

We suggested in Section~\ref{sec:general-intuition} that the number of blocks of consecutive $E$s in the van Leeuwen signature of an instance can be used to measure its difficulty. Indeed, some ``easy'' instances have few such blocks, and the instance used to prove the $\Omega(n\lg n)$ lower bound on computational complexity of optimal prefix free codes in the algebraic decision tree model in the worst case over instances of size $n$ has $n{-}1$ such blocks (the maximum possible in an instance of size $n$).  We formally define this measure as the ``alternation'' of the instance (it measures how many times the van Leeuwen algorithm ``alternates'' from an external node to an internal node) and denote it by the parameter $\alpha$:

\begin{definition}
Given an instance of the optimal prefix free code problem formed by $\nbWeights$ positive weights $\weight{1..\nbWeights}$, its \emph{alternation} $\alpha(W)\in[1..n-1]$ is the number $|{\cal S}(W)|_{EI}$ of occurrences of the substring ``$EI$'' in its van Leeuwen signature ${\cal S}(W)$.
\end{definition}

\begin{LONG}
Note that counting the number of blocks of consecutive $E$s is equivalent to counting the number of blocks of consecutive $I$s: they are the same, because the van Leeuwen signature starts with two $E$s and finished with an $I$, and each new $I$-block ends an $E$-block and vice-versa. Also, the choice between measuring the number of occurrences of ``$EI$'' or the number of occurrence of ``$IE$'' is arbitrary, as they are within a term of $1$ of each other: counting the number of occurrences of ``$EI$'' just gives a nicer range of $[1..n-1]$ (as opposed to $[0..n-2]$).
\end{LONG}
This number is of particular interest as it measures the number of iteration of the main loop in the \texttt{GDM} algorithm:

\begin{lemma}\label{result:nbLoopIterations}
Given an instance of the optimal prefix free code problem of alternation $\alpha$, 
the \texttt{GDM} algorithm performs $\alpha$ iterations of its main loop.
\end{lemma}

In the next section, we refine this result to the number of data structure operations and algebraic operations performed by the \texttt{GDM} algorithm.

\subsection{Upper Bound}\label{sec:upper-bound}

In order to measure the number of queries performed by the \texttt{GDM} algorithm, we detail how many queries are performed in each phase of the algorithm.
\begin{itemize}

\item The \emph{Initialization} corresponds to a constant number of  data structure operations: a \texttt{select} operation to find the third smallest weight, and a simple \texttt{partialSum} operation to sum the two smallest weights of the input.

\item Each \emph{Grouping} phase corresponds to a constant number of  data structure operations: a \texttt{partialSum} operation to compute the weight of the smallest internal node if needed, and a \texttt{rank} operation to identify the unpaired weights which are smaller or equal to this node.

\item The number of operations performed by each \emph{Docking} and \emph{Mixing} phase is better analyzed together: if there are $i$ symbols in the $I$-block corresponding to this phase in the van Leeuwen signature,  and if the internal nodes are grouped on $h$ levels before generating an internal node larger than the smallest unpaired weights, the \emph{Docking} phase corresponds to at most $h$ \texttt{partialSum} operations, whereas the \emph{Mixing} phase corresponds to at most $\log_2(i/2^h)$ \texttt{partialSum} operations, which develops to $\log_2(i)-h$, for a total of $\log_2 i$ data structure operations.

\item The \emph{Conclusion} phase corresponds to a number of data structure operations logarithmic in the size of the last block of $I$s in the Leeuwen's signature of the instance: in the worst case, the weight of one \emph{pure} internal node is computed for each batch, through one single \texttt{partialSum} operation each time.

\end{itemize}

Lemma~\ref{result:nbLoopIterations} and the concavity of the log yields the total number of data structure operations performed by the \texttt{GDM} algorithm:

\begin{lemma}\label{result:upper-boundQueries}
Given an instance of the optimal prefix free code problem of alternation $\alpha$, 
the \texttt{GDM} algorithm performs within $O(\alpha(1+\lg\frac{n-1}{\alpha}))$ data structure operations on the deferred data structure given as input.
\end{lemma}
\begin{proof}
For $i\in[1..\alpha]$, let $n_i$ be the number of internal nodes at the beginning of the $i$-th \emph{Docking} phase.  According to Lemma~\ref{result:nbLoopIterations} and the analysis of the number of data structure operations performed in each phase, the \texttt{GDM} algorithm performs in total within $O(\alpha + \sum_{i=1}^\alpha \lg n_i)$ data structure operations.
Since there are at most $n-1$ internal nodes, by concavity of the logarithm this is within $O(\alpha + \alpha \lg\frac{n}{\alpha})=O(\alpha(1+\lg\frac{n}{\alpha}))$. \qed
\end{proof}

Combining this result with the complexity of the \texttt{Partial Sum} deferred data structure from Lemma~\ref{result:dds} directly yields the complexity of the \texttt{GDM} algorithm in algebraic operation (and running time):

\begin{lemma}\label{result:upper-boundOperations}
Given an instance of the optimal prefix free code problem of alternation $\alpha$, the \texttt{GDM} algorithm runs in time within $O(n(1{+}\log \alpha) + \alpha(\lg n)(\lg\frac{n}{\alpha}))$, all within the algebraic computational model.
\end{lemma}
\begin{proof}
Let $q$ be the number of queries performed by the \texttt{GDM} algorithm.
Lemma~\ref{result:upper-boundQueries} implies that $q\in O(\alpha(1+\lg\frac{n}{\alpha}))$.
Plunging this into the complexity of $O(q\lg n + n\lg q)$ from  Lemma~\ref{result:dds} yields the complexity
 $O(n(1{+}\log \alpha) + \alpha(\lg n)(\lg\frac{n}{\alpha}))$. \qed
\end{proof}

Some simple functional analysis further simplifies the expression to our final upper bound:
\begin{lemma}\label{result:asymptotics}
Given two positive integers $n>0$ and $\alpha\in[1..n-1]$,
$$O(\alpha(\lg n)(\lg\frac{n}{\alpha})) \subseteq O(n(1+\lg \alpha))$$
\end{lemma}
\begin{proof}
Given two positive integers $n>0$ and $\alpha\in[1..n-1]$,
$\alpha<\frac{n}{\lg n}$ and $\frac{\alpha}{\lg \alpha}<n$.
A simple rewriting yields 
$\frac{\alpha}{\lg \alpha}<\frac{n}{\lg^2 n}$ 
and
$\alpha\lg^2 n > n\lg\alpha$ .
Then,  $n/\alpha < n$ implies 
$\alpha \times \lg n \times \lg\frac{n}{\alpha} < n\lg\alpha$, which yields the result.
\qed
\end{proof}

In the next section, we show that this complexity is indeed optimal in the algebraic decision tree model, in the worst case over instances of fixed size $n$ and alternation $\alpha$.

\subsection{Lower Bound}\label{sec:lower-bound}

A complexity within $O(n(1+\lg\alpha))$ is exactly what one could expect, by analogy with the sorting of \textsc{Multisets}: there are $\alpha$ groups of weights, so that the order within each groups does not matter much, but the order between weights from different groups matter a lot. We prove a lower bound within $\Omega(n\lg \alpha)$ by reduction to \textsc{Multiset} sorting:

\begin{lemma}\label{result:lowerBoundWorstCase}
Given the integers $n\leq 2$ and $\alpha\in[1..n{-}1]$, 
for any algorithm $A$ in the algebraic decision tree computational model,
there is a set $W[1..n]$ of $n$ positive weights of alternation $\alpha$
such that $A$ performs within $\Omega(n\lg \alpha$) operations.
\end{lemma}
\begin{proof}
For any \textsc{Multiset} $A[1..n]=\{x_1,\ldots,x_n\}$ of $n$ values from an alphabet of $\alpha$ distinct values, define the instance $W_A=\{2^{x_1},\ldots,2^{x_n}\}$ of size $n$, so that computing an optimal prefix free code for $W$, sorted by codelength, provides an ordering for $A$. $W$ has alternation $\alpha$: for any two distinct values $x$ and $y$ from $A$, the van Leeuwen algorithm pairs all the weights of value $2^x$ before pairing any weight of value $2^y$, so that the van Leeuwen signature of $W_A$ has $\alpha$ blocks of consecutive $E$s. The lower bounds then results from the classical lower bound on sorting \textsc{Multisets} in the comparison model in the worst case over \textsc{Multisets} of size $n$ with $\alpha$ distinct symbols.\qed
\end{proof}

\begin{INUTILE}
\begin{proof}

We prove a lower bound within $\Omega(n\lg \alpha)$ by defining for any given value of $n$ and $\alpha$ an instance with $\alpha$ groups of $n/\alpha$ equal weights, so that computing an optimal prefix free code for this instance corresponds to sorting $n$ values from an alphabet of $\alpha$ distinct symbols.

Let $k>\frac{n}{\alpha}$ be a constant strictly larger than $n/\alpha$.  We define the instance $W_q$ composed by $\alpha$ groups of equal weights, where for $i\in[1..\alpha]$ the $i$-th group is composed of $\frac{n}{\alpha}$ weights of value $2^{k^{i-1}}$:
$$W_k = \{1,\ldots,1,2^k,\ldots,2^k,2^{k^2},\ldots,2^{k^2},\ldots, 2^{k^\alpha},\ldots,2^{k^\alpha}\}$$
(...)

\end{proof}
\end{INUTILE}

\begin{TODO}
We can even show a stronger result:

\begin{lemma}\label{result:lowerBoundAverage}
Given the integers $n\leq 2$ and $\alpha\in[1..n{-}1]$, for any algorithm $A$ in the algebraic decision tree computational model, there is a set $W[1..n]$ of $n$ positive weights of alternation $\alpha$ such that $A$ performs within $\Omega(n\lg \alpha)$ operations on average on random permutations of $W$.
\end{lemma}
\end{TODO}

\begin{INUTILE}
\section{Experimentations}\label{sec:experimentations}
\subsection{Alternation in Practice}\label{sec:alternation-practice}
\subsection{Running Time in Practice}\label{sec:runn-time-pract}
\end{INUTILE}

We compare our results to previous results in the next section.
\pagebreak[3]
\section{Discussion}
\label{sec:discussion}

\begin{LONG}
We described an algorithm computing an optimal prefix free code for $n$ unsorted positive weights in time within $O(n(1{+}\lg \alpha))\subseteq O(n\lg n)$, where the alternation $\alpha\in[1..n{-}1]$ roughly measures the amount of sorting required by the computation, by combining van Leeuwen's results about optimal prefix free codes~\cite{1976-ICALP-OnTheConstructionOfHuffmanTrees-Leeuwen}, known since 1976, with results about Karp\etal's results about Deferred Data Structures~\cite{1988-JC-DeferredDataStructuring-KarpMotwaniRaghavan}, known since 1988.
\end{LONG}
The results described above yields many new questions, of which we discuss only a few in the following sections.

\begin{VLONG}
We discuss in this section how those results relate to previous results on optimal prefix free codes (Section~\ref{sec:relat-prev-work}), to other results on Deferred Data Structures obtained since 1988 (Section~\ref{sec:aplic-dynamic-results} and~\ref{sec:aplic-refined-results})\begin{LONG}, to the lack of practical applications of our results on optimal prefix free codes (Section~\ref{sec:impact-our-results}), \end{LONG} and about perspectives of research on this topic (Section~\ref{sec:perspectives}).
We list in the appendix~\ref{sec:relev-optim-pref} some interesting quotes about the importance of optimal prefix free codes in general.
\end{VLONG}

\subsection{Relation to previous work on optimal prefix free codes} \label{sec:relat-prev-work} 

In 2006, Belal\etal~\cite{2006-STACS-DistributionSensitiveConstructionOfMinimumRedundancyPrefixCodes-BelalElmasry}, described a variant of Milidi{\'u}\etal's algorithm~\cite{2001-IEEE-ThreeSpaceEconomicalAlgorithmsForCalculatingMinimumRedundancyPrefixCodes-MilidiuPessoaLaber,1998-TR-ASpaceEconomicalAlgorithmForMinimumRedundancyCoding-MiliduPessoaLaber} to compute optimal prefix free codes, announcing that it performed $\Oh(\nbCodeLengths\nbWeights)$ algebraic operations when the weights are not sorted, where $\nbCodeLengths$ is the number of distinct code lengths in some optimal prefix free code.

\begin{LONG}
They describe an algorithm claimed to run in time $\Oh({16}^\nbCodeLengths\nbWeights)$ when the weights are unsorted, and propose to improve the complexity to $\Oh(\nbCodeLengths\nbWeights)$ by partitioning the weights into smaller groups, each corresponding to disjoint intervals of weights value\footnote{Those results were downgraded in the December 2010 update of their initial 2005 publication through \texttt{Arxiv}~\cite{2005-ARXIV-DistributionSensitiveConstructionOfMinimumRedundancyPrefixCodes-BelalElmasry}.}.  The claimed complexity is asymptotically better than the one suggested by Huffman when $\nbCodeLengths\in o(\log\nbWeights)$, and they raise the question of whether there exists an algorithm running in time $\Oh(\nbWeights\log\nbCodeLengths)$.
\end{LONG}

Like the \texttt{GDM} algorithm, the algorithm described by Belal\etal~\cite{2006-STACS-DistributionSensitiveConstructionOfMinimumRedundancyPrefixCodes-BelalElmasry} for the unsorted case is based on several computations of the median of the weights within a given interval, in particular, in order to select the weights smaller than some well chosen value.  The essential difference between both work is the use of deferred data structure, which simplifies both the algorithm and the analysis of its complexity.

\begin{LONG}
\subsection{Applicability of dynamic results on Deferred Data Structures}
\label{sec:aplic-dynamic-results}

Karp\etal~\cite{1988-JC-DeferredDataStructuring-KarpMotwaniRaghavan}, when they defined the first Deferred Data Structures, supporting \texttt{rank} and \texttt{select} on \textsc{Multisets} \begin{VLONG} and other queries on \textsc{Convex Hull}\end{VLONG}, left as an open problem the support of dynamic operators such as \texttt{insert} and \texttt{delete}: Ching\etal~\cite{1990-IPL-DynamicDeferredDataStructuring-ChingMehlhornSmid} quickly demonstrated how to add such support in good amortized time.

The dynamic addition and deletion of elements in a deferred data structure (added by Ching\etal~\cite{1990-IPL-DynamicDeferredDataStructuring-ChingMehlhornSmid} to Karp\etal~\cite{1988-JC-DeferredDataStructuring-KarpMotwaniRaghavan}'s results) does not seem to have any application to the computation of optimal prefix free codes: even if the list of weights was dynamic, further work is required to build a deferred data structure supporting prefix free code queries.
\end{LONG}

\begin{LONG}
\subsection{Applicability of refined results on Deferred Data Structures}
\label{sec:aplic-refined-results}

Karp~{\etal}'s analysis~\cite{1988-JC-DeferredDataStructuring-KarpMotwaniRaghavan} of the complexity of the deferred data structure is in function of the total number $q$ of queries and operators, while Kaligosi\etal~\cite{2005-ICALP-TowardsOptimalMultopleSelection-KaligosiMehlhornMunroSanders} analyzed the complexity of an offline version in function of the size of the gaps between the positions of the queries.  Barbay\etal\cite{2013-ESA-OnlineRankSelect-BarbayGuptaJoRaoSorenson} combined the three results into a single deferred data structure for \textsc{Multisets} which supports the operators \texttt{rank} and \texttt{select} in amortized time proportional to the entropy of the distribution of the sizes of the gaps between the positions of the queries.

At first view, one could hope to generalized the refined entropy analysis (introduced by Kaligosi\etal~\cite{2005-ICALP-TowardsOptimalMultopleSelection-KaligosiMehlhornMunroSanders} and applied by Barbay\etal\cite{2013-ESA-OnlineRankSelect-BarbayGuptaJoRaoSorenson} to the online version) of \textsc{Multisets} deferred data structures supporting \texttt{rank} and \texttt{select} to the computational complexity of optimal prefix free codes: a complexity proportional to the entropy of the distribution of codelengths in the output would nicely match the lower bound of $\Omega(\nbCodeLengths(1+\entropy(n_1,\ldots,n_h)))$ suggested by information theory, where the output contains $n_i$ codes of length $l_i$, for some integer vector $(l_1,\ldots,l_h)$ of distinct codelengths and some integer $h$ measuring the number of distinct codelengths. Our current analysis does not yield such a result: the gap lengths between queries in the list of weights are not as regular as $(l_1,\ldots,l_h)$.
\end{LONG}

  \begin{LONG}
  \subsection{Potential Practical Impact of our Results}
  \label{sec:impact-our-results}

  The impact of our faster algorithm on the execution time of optimal prefix free code based techniques should definitely be evaluated further.
  Yet, we expect it to be of little importance in most cases: compressing a sequence $S$ of $|S|$ messages from an input alphabet of size $\nbWeights$ requires not only computing the code (in time $\Oh(\nbWeights)$ using our solution), but also computing the weights of the messages (in time $|S|$), and encoding the sequence $S$ itself using the computed code (in time $\Oh(|S|)$).
  Improving the code computation time will improve on the compression time only in cases where the size $\nbWeights$ of the input alphabet is very large compared to the length $|S|$ of the compressed sequence.
  One such application is the compression of texts in natural language, where the input alphabet is composed of all the natural words~\cite{2000-TOIS-FastAndFlexibleWordSearchingOnCompressedText-MouraNavarroZivianiBaezaYates}.
  Another potential application is the boosting technique from Ferragina\etal~\cite{2005-JACM-BoostingTextualCompressionInOptimalLinearTime-FerraginaGiancarloManziniSciortino}, which divides the input sequence into very short subsequence and computes a prefix free code for each subsequences on the input alphabet of the whole sequence.
  \end{LONG}

\subsection{Perspectives}\label{sec:perspectives}

One could hope for an algorithm which complexity would {match the lower bound} of $\Omega(\nbCodeLengths(1+\entropy(n_1,\ldots,n_h)))$ suggested by information theory, where the output contains $n_i$ codes of length $l_i$, for some integer vector $(l_1,\ldots,l_h)$ of distinct codelengths and some integer $h$ measuring the number of distinct codelengths. Our current analysis does not yield such a result: the gap lengths between queries in the list of weights are not as regular as $(l_1,\ldots,l_h)$, but a refined analysis might.  Minor improvements of our results could be brought by studying the problem in {external memory}, where deferred data structures have also been developed~\cite{2006-JALG-ExternalSelection-Sibeyn,2014-WALCOM-DynamicOnlineMultiSelectionInInternalAndExternalMemory-BarbayGuptaJoRaoSorenson}, or when the {alphabet size} is {larger than two}, as in the original article from Huffman~\cite{1952-IRE-AMethodForTheInstructionOfMinimumRedundancyCodes-Huffman}.

Another promising line of research is given by {variants of the original problem}, such as \textsc{Optimal Bounded Length Prefix Free Codes}, where the maximal length of each word of the prefix free code must be less than or equal to a parameter $l$, while still minimizing the entropy of the code; or such as the \textsc{Order Constrained Prefix Free Codes}, where the order of the words of the codes is constrained to be the same as the order of the weights. Both problems have complexity $O(\nbWeights\lg\nbWeights)$ in the worst case over instances of fixed input size $\nbWeights$, while having linear complexity when all the weights are within a factor of two of each other, exactly as in the original problem.

\begin{LONG}
A logical step would be to study, among the communication solutions using an optimal prefix free code computed offline, which can now afford to compute a new optimal prefix free code more frequently and see their compression performance improved by a faster prefix free code algorithm.
\end{LONG}

\begin{LONG}
Another logical step would be to study, among the compression algorithms computing an optimal prefix free code on each new instance (e.g. \texttt{JPEG}, \texttt{MP3}, \texttt{MPEG}), which ones get a better their time performance by using a faster prefix free code algorithm.  
\end{LONG}

\begin{LONG}
\appendix

\section{Appendix}
\subsection{Relevance of Prefix Free codes in General}
    \label{sec:relev-optim-pref}

    Albeit 60 year old, Huffman's result is still relevant nowadays.
    Optimal Prefix Free codes are used not only for compressed
    encodings: they are also used in the construction of compressed
    data structures for
    permutations~\cite{2009-STACS-CompressedRepresentationsOfPermutationsAndApplications-BarbayNavarro},
    and using similar techniques for sorting faster multisets which
    contains subsequences of consecutive positions already
    ordered~\cite{2009-STACS-CompressedRepresentationsOfPermutationsAndApplications-BarbayNavarro}.

    In 1991, Gary Stix~\cite{1991-SAME-ProfileDavidAHuffman-Stix}
    stated that ``\emph{Large networks of IBM computers use it. So do
      high-definition television, modems and a popular electronic
      device that takes the brain work out of programming a
      videocassette recorder. All these \textbf{digital wonders rely
        on} the results of \textbf{a 40-year-old term paper by a
        modest} Massachusetts Institute of Technology \textbf{graduate
        student}-a data compression scheme known as Huffman encoding
      (...)  \textbf{Products that use Huffman code might fill a
        consumer electronics store}. A recent entry on the shop shelf
      is VCR Plus+, a device that automatically programs a VCR and is
      making its inventors wealthy. (...)  Instead of confronting the
      frustrating process of programming a VCR, the user simply types
      into the small handheld device a numerical code that is printed
      in the television listings. When it is time to record, the
      gadget beams its decoded instructions to the VCR and cable box
      with an infrared beam like those on standard remote-control
      devices. This turns on the VCR, sets it (and the cable box) to
      the proper channel and records for the designated time}.''.

    In 1995,
    Moffat and Katajainen~\cite{1995-WADAS-InPlaceCalculationOfMinimumRedundancyCodes-MoffatKatajainen},
    stated that: ``\emph{The algorithm introduced by Huffman for
      devising minimum-redundancy prefix free codes is well known and
      continues to enjoy \textbf{widespread use in data compression
        programs}. Huffman's method is also a good illustration of the
      greedy paradigm of algorithm design and, at the implementation
      level, provides a useful motivation for the priority queue
      abstract data type. For these reasons Huffman's algorithm enjoys
      \textbf{a prominence enjoyed by only a} relatively \textbf{small
        number of fundamental methods}}''.

    In 1997, Moffat and
    Turpin~\cite{1997-IEEE-OnTheImplementstionOfMinimumRedundsncyPrefixCodes-MoffatTurpin}
    stated that those were ``\emph{one of the \textbf{enduring
        techniques of data compression}. It was used in the venerable
      PACK compression program, authored by Szymanski in 1978, and
      \textbf{remains no less popular today}}''.

    In 2010 Donald E. Knuth was
    quoted~\cite{2010-BOOK-DiscreteMathematics-Chandrasekaran} as
    saying that: ``\emph{Huffman code is one of the
      \textbf{fundamental ideas that people} in computer science and
      data communications \textbf{are using all the time}}''.

    In 2010, the answer to the question ``\emph{What are the
      real-world applications of Huffman coding?}'' on the website
    \texttt{Stacks
      Exchange}~\cite{2010-stacksExchange-realWorldApplicationsHuffman}
    states that ``\emph{Huffman is widely used in all the mainstream
      compression formats that you might encounter - from GZIP, PKZIP
      (winzip etc) and BZIP2, to image formats such as JPEG and
      PNG.}''.

    The Wikipedia website on Huffman coding states that
    ``\emph{Huffman coding today is often used as a "back-end" to some
      other compression method. DEFLATE (PKZIP's algorithm) and
      multimedia codecs such as JPEG and MP3 have a front-end model
      and quantization followed by Huffman
      coding.}''~\cite{2012-wikipedia-HuffmanCoding}.

    Ironically, the pseudo-optimality of this algorithm seems to have
    become part of the folklore of the area, as illustrated by a quote
    from
    Parker\etal~\cite{1999-SIAM-HuffmanCodesSubmodularOptimization-ParkerRam}
    in 1999: ``\emph{While there may be little hope of improving on
      the $\Oh(\nbWeights\log\nbWeights)$ complexity of the Huffman
      algorithm itself, there is still room for improvement in our
      understanding of the algorithm}.''.
  \end{LONG}
